\documentclass[conference]{IEEEtran}
\usepackage{cite}
\usepackage{algorithm,algorithmicx}
\usepackage[noend]{algpseudocode}
\usepackage{graphics}
\usepackage{epsfig}
\usepackage{amsmath,amssymb,amsfonts}
\usepackage{amsthm}
\usepackage{graphicx}
\usepackage{textcomp}
\usepackage{xcolor}
\usepackage{flushend}
\usepackage{threeparttable}
\usepackage{verbatim}
\usepackage{url}
\usepackage{verbatim}
\usepackage{bm}
\def\BibTeX{{\rmB\kern-.05em{\sci\kern-.025emb}\kern-.08emT\kern-.1667em\lower.7ex\hbox{E}\kern-.125emX}}

\newtheorem{theorem}{Theorem}
\newtheorem{lemma}[theorem]{Lemma}

\newtheorem{claim}[theorem]{Claim}

\IEEEoverridecommandlockouts

\begin{document}

\title{CycLedger: A Scalable and Secure Parallel Protocol for Distributed Ledger via Sharding\\
\thanks{This research reported in this work was supported by the National Natural Science Foundation of China (No. 61761146005, 61632017).}
}

\author{\IEEEauthorblockN{Mengqian Zhang\IEEEauthorrefmark{1}, Jichen Li\IEEEauthorrefmark{2}, Zhaohua Chen\IEEEauthorrefmark{3}, Hongyin Chen\IEEEauthorrefmark{4} and Xiaotie Deng\IEEEauthorrefmark{5}}
\IEEEauthorblockA{\IEEEauthorrefmark{1}School of Electronic Information and Electrical Engineering, Shanghai Jiao Tong University, Shanghai, 200240, China}
\IEEEauthorblockA{\IEEEauthorrefmark{2}\IEEEauthorrefmark{3}\IEEEauthorrefmark{4}\IEEEauthorrefmark{5}School of Electronics Engineering and Computer Science, Peking University, Beijing, 100871, China}
\IEEEauthorblockA{Email: \IEEEauthorrefmark{1}mengqian@sjtu.edu.cn, \{\IEEEauthorrefmark{2}1600012970, \IEEEauthorrefmark{3}chenzhaohua, \IEEEauthorrefmark{4}1600012903, \IEEEauthorrefmark{5}xiaotie\}@pku.edu.cn}

\IEEEauthorblockA{\IEEEauthorrefmark{1}\IEEEauthorrefmark{2}\IEEEauthorrefmark{3}\IEEEauthorrefmark{4}These authors contributed equally to this work.}
}
\maketitle

\begin{abstract}
     Traditional public distributed ledgers have not been able to scale-out well and work efficiently. Sharding is deemed as a promising way to solve this problem. By partitioning all nodes into small committees and letting them work in parallel, we can significantly lower the amount of communication and computation, reduce the overhead on each node's storage, as well as enhance the throughput of the distributed ledger. Existing sharding-based protocols still suffer from several serious drawbacks. The first thing is that all non-faulty nodes must connect well with each other, which demands a huge number of communication channels in the network. Moreover, previous protocols have faced great loss in efficiency in the case where the honesty of each committee's leader is in question. At the same time, no explicit incentive is provided for nodes to actively participate in the protocol.
     
     We present CycLedger, a scalable and secure parallel protocol for distributed ledger via sharding. Our protocol selects a leader and a partial set for each committee, who are in charge of maintaining intra-shard consensus and communicating with other committees, to reduce the amortized complexity of communication, computation, and storage on all nodes. We introduce a novel semi-commitment scheme between committees and a recovery procedure to prevent the system from crashing even when leaders of committees are malicious. To add incentive for the network, we use the concept of reputation, which measures each node's trusty computing power. As nodes with a higher reputation receive more rewards, there is an encouragement for nodes with strong computing ability to work honestly to gain reputation. In this way, we strike out a new path to establish scalability, security, and incentive for the sharding-based distributed ledger.
\end{abstract}


\section{Introduction}\label{Introduction}

The blockchain revolution has established a milestone with Nakamoto's Bitcoin~\cite{Bit} during the long search for a dreaming digital currency. It maintains the distributed database in a decentralized manner and has achieved a certain maturity to provide the electronic community with a service that captures the most important features of cash: secure, anonymity, easy to carry, easy to change hand, and exchangeable across a nation's boundaries. At the same time, Bitcoin realizes the tamperproof of transactions. 

Unfortunately, the transaction throughput of Bitcoin is 3 to 4 orders of magnitude away from those centralized payment-processors like Visa, which stops its popularization. Specifically, Bitcoin can process only 3.3-7 transactions per second in 2016~\cite{Croman2016} while Visa handles more than 24,000 transactions per second~\cite{Visa}. The low scalability of Bitcoin is the key factor causing this displeasing result.

Sharding is a famous technique in building the “scale-out” database by separating the whole state into multiple parts and making them work concurrently. Inspired by this idea, sharding is also used to benefit distributed ledgers in these years. When nodes are partitioned into groups, less computational power is wasted and higher processing capacity is achieved. However, most sharding-based protocols fail to maintain high efficiency when committee leaders betray~\cite{Omni}~\cite{Rapid}. They also fail to give an explicit incentive for nodes to join the protocol and behave honestly.

We present CycLedger, a fully decentralized payment-processor that provides scalability over the number of nodes and provable security. Specifically, we expect our protocol to remain robust when leaders in each committee are faulty. This is a problem that remains unsolved until now. Meanwhile, we hope there is enough encouragement for a node to participate honestly in our protocol. Furthermore, scalability is realized via sharding nodes into concurrent committees as previous works have shown~\cite{Omni}~\cite{Ela}~\cite{Rapid}. 

We assume there are no more than 1/3 malicious nodes at any time during the execution of our protocol, which is the best result achieved so far in sharding blockchain schemes~\cite{Rep}~\cite{Rapid}. Meanwhile, an adversary can corrupt any nodes as he/she likes, but it requires some time for such corruption attempts to take effect. We select its members in each committee pseudorandomly so that with overwhelming probability, there are more than 1/2 honest validators in a committee. Furthermore, we construct a partial set in each committee. Members in this set supervise the leader's behavior and act as alternates when the latter is confirmed to be vicious. By setting the size of each partial set to an appropriately large number, we ensure that only with negligible probability there are no honest nodes in a partial set. Finally, we design a recovery procedure to select a new leader when the original one is dishonest. Concretely, when a leader is found to violate the restriction of the protocol, he/she will be evicted and a node in the partial set will take his/her place. By introducing digital signature and semi-commitments between committees, we can guarantee that a non-void block can be successfully generated in each round with high probability.

To provide validators with enough incentive to enter our network, we introduce the concept of reputation to CycLedger. We hope one's reputation could be a good reflection of his/her computing power so that when we distribute the total revenue according to nodes' reputation, those honest nodes who contribute more to access transactions are paid more. To achieve this goal, we consider one's vote on a certain set of transactions. The more similar his/her vote is with the final decision, the more reputation he/she gains. In this way, as nodes with a higher computational resource can process more transactions correctly within a given time, they will obtain a higher accumulated reputation, consequently, winning more reward than other nodes.

The convenience reputation provides us is that we can simply sort out those nodes with strong computing ability.  In CycLedger, those nodes with the best reputation are selected as leaders of each round, hoping they can use their abundant computational resources to bring more transactions into a block. In this way, a promotion on throughput is expected to see in CycLedger.

This paper is organized as follows. In Section~\ref{Related}, we review previous works that are relevant to our protocol. In Section~\ref{Model}, we state the network and threat models we use, define the problems we aim to solve and give an overview of CycLedger. We elaborate on our main protocol in Section~\ref{Protocol}. After that, we give the analysis on the security and incentive of our protocol in Section~\ref{Security} and Section~\ref{Incentive}, respectively. Finally, we conclude the paper in Section~\ref{Conclusion}. 

\section{Background and Related Work}\label{Related}

\begin{table*}[!t]
    \renewcommand{\arraystretch}{1.35}
    \caption{A Comparison of CycLedger with Previous Sharding Blockchain Protocols}
    \label{Comparison}
    \centering
    \begin{threeparttable}
       \begin{tabular}{|c|c|c|c|c|}
        \hline
        & Elastico & OmniLedger & RapidChain & CycLedger \\
        \hline
        Resiliency & $t < n/4$ & $t < n/4$ & $t < n/3$ & $t < n/3$ \\
        \hline
        Complexity & $\Omega(n)$ & $O(n)$ & $O(n)$ & $O(n)$ \\
        \hline
        Storage & $O(n)$ & $O(c + \log(m))$ & $O(c)$ & $O(m^2 / n + c)$ \\
        \hline
        Fail Probability within a Round & $\Omega(me^{-c/40})$ & $O(me^{-c/40})$ & $me^{-c/12} + (1/2)^{27}$ & $m(e^{-c/12} + (1/3)^\lambda)$ \\
        \hline
        Decentralization & no always-honest party & an honest client & an honest reference committee & no always-honest party \\
        \hline
        High Efficiency w.r.t Dishonest Leaders & $\times$ & $\times$ & $\times$ & \checkmark \\
        \hline
        Incentives & $\times$ & $\times$ & $\times$ & \checkmark \\
        \hline
         Burden on Connection & heavy & heavy & heavy & light \\
        \hline
        \end{tabular}
        \begin{tablenotes}
           \footnotesize
           \item[1] $n$: total amount of nodes in the network\ $m$: amount of committees\ $c$: amount of nodes in a committee, we mention that $n = mc$.
           \item[2] $\lambda$: size of a partial set. Usually $\lambda$ is set to be no less than 40. 
           \item[3] When computing complexity, it is assumed that a cross-shard transaction be relative with all committees.
        \end{tablenotes}
    \end{threeparttable}
\end{table*}

\subsection{Sharding-based Blockchains}

In traditional blockchain protocols, all nodes in the network have to agree on a certain set of transactions. This scheme leads to a very low throughput as only a rather fixed amount of transactions are included in a block regardless of the number of nodes in the network. An alternative way is to partition nodes into parallel committees and let each committee maintain the status of a certain group of users. Under this method, the amount of transactions is proportional to the number of committees rather than a constant. This technique is known as \textit{sharding} and is considered as a good way to help blockchain scale well in literature.

RSCoin~\cite{RS} implements a centralized monetary system via sharding, however, it is not decentralized and cannot transplant to distributed ledgers. Elastico~\cite{Ela} is the first sharding-based protocol for public blockchains which can tolerate up to a fraction of 1/4 of malicious parties. Unfortunately, it has a very weak security guarantee as the randomness in each epoch of the protocol can be biased by the \textit{Adversary}. Meanwhile, Elastico's small committees (only about 100 nodes in a committee) cause a high probability to fail under a 1/4 adversary, and cannot be released in a PoW system~\cite{Is}. Specifically, when there are 16 shards, the failure probability is 97\% over only 6 epochs~\cite{Omni}. 
OmniLedger~\cite{Omni} also allows the adversary to take control of at most 25\% of the validators as well as assuming the adversary to be mildly-adaptive, nevertheless, it depends on the assumption that there is a never-absent trusty client to schedule the leaders' interaction when handling cross-shard transactions. RapidChain~\cite{Rapid} enhances the efficiency of sharding-based blockchain protocols on a large scale, but the protocol guarantees high efficiency only when leaders of each committee are honest, an unrealistic assumption in practice. Concretely, in expectation, there is a proportion of 1/3 leaders that are malicious in a round. Under this condition, cross-shard transactions may hardly be included in a block. Furthermore, the protocol does not have an explicit incentive for nodes to participate in. At the same time, all the above postulate a good connection between any pair of truthful nodes, which causes a huge burden in creating connection channels.

We give a comparison of CycLedger with previous sharding blockchain protocols on several aspects in table \ref{Comparison}.

\subsection{Distributed Randomness and Cryptographic Sortition}

It has long been a critical task in blockchain to come up with a random beacon each round (or step, slot, epoch). Algorand~\cite{Algo} makes use of Verifiable Random Functions (VRFs)~\cite{VRF}. The seed of the current round, the round number as well as the leader's secret key are required to produce the next round's seed and a proof of it. In this way, when the leader is honest, the seed is pseudorandom and unpredictable. However, the beacon can be somehow biased when the leader is malicious. Ouroboros~\cite{Ouro} takes the usage of the Publicly Verifiable Secret Sharing (PVSS) scheme~\cite{PVSS2}~\cite{PVSS1} to generate the random seed. This scheme is now known as a secure way to distributedly generate a random number and is also used in OmniLedger~\cite{Omni}. However, OmniLedger deploys RandHound~\cite{RandHound} to generate the Common Reference String (CRS). Therefore, the protocol can only allow a 1/3 proportion of malicious nodes in each committee. At the same time, a trusty client is also required in RandHound. Other protocols~\cite{Snow}~\cite{OuroPraos} use an external cryptographic hash function (CRHF), which takes an unpredictable and tamper-resistant value (\textit{e.g.}, some values contained in previous blocks) as input. The output of the function is seen as the randomness of an epoch. 

The technique of cryptographic sortition is widely used in blockchain protocols to select a small set of validators pseudorandomly. In Algorand~\cite{Algo}, a node's VRF value on a given input is the ticket to enter the committee. At the same time, VRF also provides proof so that everyone can verify if a node is selected. We mention that many protocols~\cite{Snow}~\cite{OuroPraos} now use this scheme to ensure pseudorandomness, verifiability, and unpredictability without loss of efficiency. Other protocols are utilizing different mechanisms. For example, the Sleepy Model~\cite{Sleepy} only uses pseudorandom functions while Ouroboros~\cite{Ouro} randomly generates several coins and uses the Follow-the-Satoshi (FTS) algorithm~\cite{Card} to find their owners who are deemed as leaders. However, these methods are either lack of security or not quite efficient.

\subsection{Reputation and Blockchain}

The open, anonymous and decentralized environment of Peer-to-Peer (P2P) systems brings significant advantages, such as the strong expansibility and load balancing. But in the meantime, it provides great convenience for the self-interested and malicious users to take strategic behaviors, which could lead to the inefficiency or even failure of the system. For example, current P2P file-sharing networks suffer from the rifeness of inauthentic contents~\cite{FilePoll}. \textit{Reputation}, regarded as a tool to record the action of users, has been widely used in the collaborative P2P applications to solve this problem. By designing a proper metric and corresponding algorithm, it could help the aforementioned file-sharing system effectively to identify malicious peers~\cite{EigenTrust}, avoid undesirable contents~\cite{Walsh2006} ~\cite{Damiani2002} and drive cooperation among strategic users~\cite{Feldman2004}. 

Inspired by the success of reputation in many other P2P systems, some researches try to introduce it into the area of blockchain. Nojoumian \textit{et al.}~\cite{Nojoumian2018} proposes a new framework for the PoW-based blockchain, in which each miner has a public reputation value reflecting his historical mining behavior. One's reputation influences his opportunity to participate in future minings, thus, for their long-term interests, miners are encouraged to avoid dishonest mining strategies in this model. Repchain~\cite{Rep} first introduces the reputation mechanism into the sharding-based blockchain system. Reputation is designed to characterize validators' trust and activeness, based on their decisions on the transactions list. Repchain uses the accumulated reputation values for balanced sharding, leader selections, and transaction fee distributions. It is claimed that the schema provides a high incentive for validators to work hard and honestly, as well as improving the system performance. However, building a committee with the help of a reputation reduces its randomness. That is, it indeed trades security for better incentives.

\section{Model, Problem Definition and Overview}\label{Model}

\subsection{Notation}

CycLedger works in rounds. In each round $r$, suppose there are $n$ ($n$ is changing as $r$ changes, but to simplify the notation, we use $n$ instead of $n^r$) nodes in the network and each of them has a reputation $w^r_1, w^r_2, \cdots, w^r_n$ which changes as the protocol executes. We use a Public-Key Infrastructure (PKI) to give each node a public/secret key pair ($PK, SK$). An honest-majority referee committee $C_R^r$ is selected in round $r - 1$ to manage nodes' identities, produce next round's randomness $R^{r+1}$ and propose round $r$'s block $B^r$. At the same time, all other nodes are partitioned into $m$ committees which we denote as $C_1^r, C_2^r, \cdots, C_m^r$. Note that we require each committee a trusty majority. Therefore, in our protocol, we would like the size of each committee to be $c = O(\log^2 n)$ in expectation. Readers should mention that $n = mc$. Each committee $C_i^r$, excluding the referee committee, includes a leader $l_i^r$, $\lambda$ potential leaders and $c - \lambda - 1$ members. Potential leaders in $C_i^r (1 \leq i \leq m)$ form the partial set of the committee, the latter is denoted as $C_{i, partial}^r = \{c_{i, 1}^r, c_{i, 2}^r, \cdots, c_{i, \lambda}^r\}$. It is ensured that there is at least one non-faulty node in each partial set. In practice, $\lambda$ can be an appropriate value, like 40. Fig.~\ref{fig1} demonstrates the hierarchical structure of different nodes and committees.

\begin{figure*}[!t]
\centering
\includegraphics[width=0.72\textwidth]{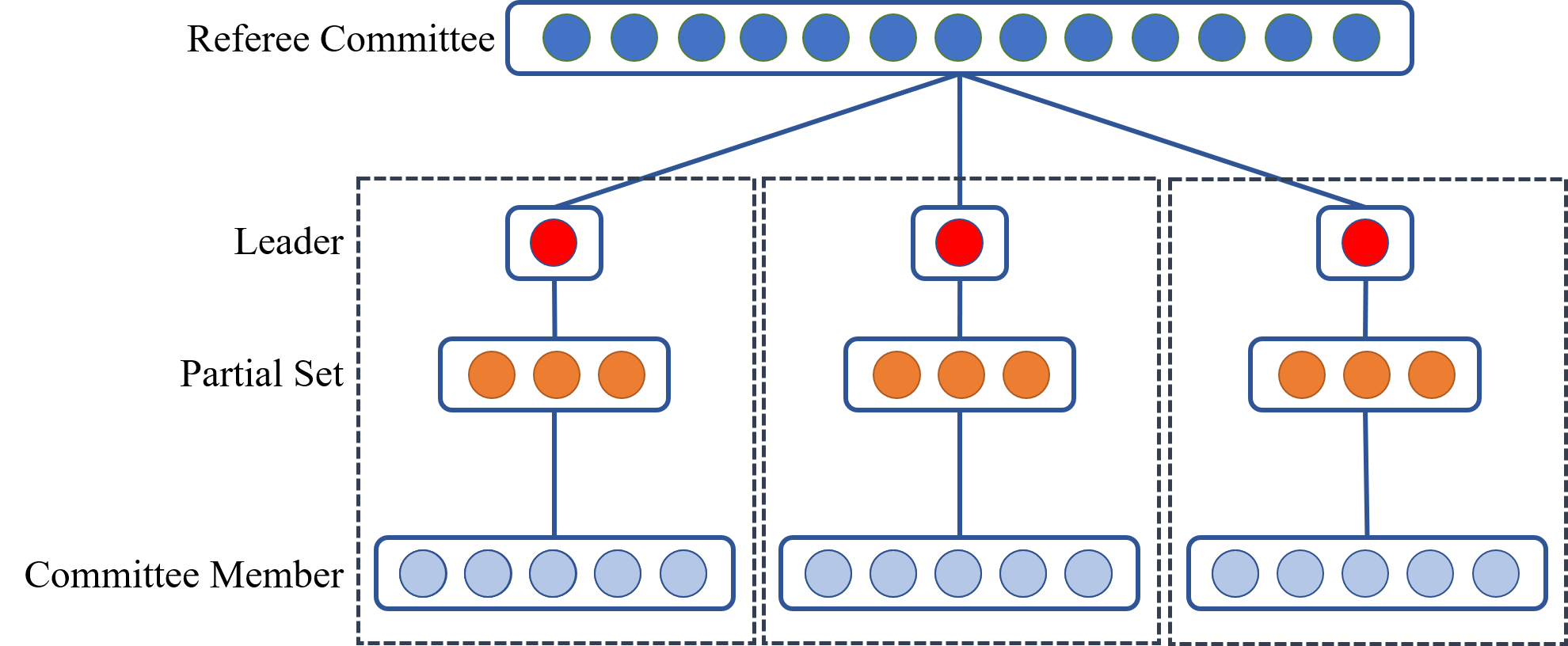}
\caption{Hierarchical structure.}
\label{fig1}
\end{figure*}
 
Besides, with overwhelming probability, each round is successfully terminated (\textit{i.e.}, all expected operations are finished) within a fixed time $T$.

\subsection{Network Model}

We assume the good connection within a committee while all leaders and partial set members are linked. Furthermore, we suppose that each leader or partial set member is connected with the whole referee committee $C_R$. This requires a far less amount of reliable connection channels than other works~\cite{Rep}~\cite{Omni}~\cite{Ela}~\cite{Rapid} in which they require a good connection among all honest nodes. Meanwhile, like existing works~\cite{Omni}~\cite{Rapid}, we assume synchronous communication within committees  (\textit{i.e.}, the delay of transporting every message is within some $\Delta$) which is realistic in real-world as a committee only consists of several hundred nodes. Meanwhile, all leaders and partial set members are synchronously linked, however, with a larger time delay of $\Gamma$. Concerning other connections, we only need to assume partially-synchronous channels~\cite{IRSDP}~\cite{PBFT}~\cite{Omni}.

\subsection{Threat Model}

As we use digital signatures in reaching consensus as RapidChain~\cite{Rapid} and RepChain~\cite{Rep} do, we may assume the existence of a probabilistic polynomial-time \textit{Adversary} which takes control of less than 1/3 part of total nodes. Corrupted nodes may collude and act out arbitrary behaviors like sending wrong messages or simply pretending to be offline. The adversary can change the order of messages sent by non-faulty nodes for the restriction given in our network model, just like in classical BFT protocols~\cite{IRSDP}. Other nodes, known as \textit{honest} nodes, always follow the protocol and do nothing exceeding the regulation. At the same time, we allow the adversary to be mildly-adaptive, which means that he/she is allowed to corrupt a set of nodes at the start of any round. Nevertheless, such corruption attempts require at least a round's time to take effect. Also, we assume all nodes in the network have access to an external random oracle $H$ which is collision-resistant as well as a Verifiable Random Function (VRF) scheme.

\subsection{Problem Definition}

 We assume that a large set of transactions are continuously sent to our network by external \textit{users}. Users are almost equally divided into $m$ shards. The status of each shard, including the users' identity and Unspent Transaction Outputs (UTXOs), is maintained by the corresponding committee. All processors have access to an authentication function $V$ to verify whether a transaction is legitimate, \textit{e.g.}, the sum of all inputs of the transaction is no less than the sum of all outputs and there is no double-spending. Our goal is similar to ~\cite{Ela}~\cite{Rapid} but slightly different. We seek for a protocol $\Pi$ which, with a given set of transactions as input, outputs a subset, $TX$, such that the following properties hold:
\begin{itemize}
    \item \textit{Security}. In each round, the protocol fails at a probability no more than $2^{-\lambda}$, where $\lambda$ is the security parameter.
    \item \textit{Validity}. Each transaction in $TX$ passes the verification, \textit{i.e.}, $\forall tx \in TX, V(tx) = True$.
    \item \textit{Scalability}. $|TX|$ grows quasi-linearly with $n$.
    \item \textit{Incentive}. Nodes are awarded to execute the protocol within the given restriction.
\end{itemize}

We mention that CycLedger differs from previous protocols ~\cite{Omni}~\cite{Ela}~\cite{Rapid} as there is an explicit incentive design in it.

\subsection{Protocol Overview}

Roughly, in each round $r$, our protocol consists of the following phases.

\begin{itemize}
    \item \textbf{Committee Configuration}. As the referee committee $C_R^r$, leaders and partial sets of round $r$ are already determined in the previous round, in this phase, all other verified nodes are grouped and committees are formed.
    \item \textbf{Semi-Commitment Exchanging}. Each leader constructs a semi-commitment of all members in his/her committee and sends it to $C_R^r$, the partial set of his/her committee and all other leaders. Loosely speaking, the semi-commitment helps to detect a malicious leader when he/she tries to cheat in the phase of inter-committee consensus.
    \item \textbf{Intra-committee Consensus}. In this phase, non-faulty nodes within a committee reach an agreement on those transactions whose inputs and outputs only relate with the shard they are responsible for.
    \item \textbf{Inter-committee Consensus}. In short, all truthful nodes agree on the validity of cross-shard transactions, whose inputs and outputs are related to multiple shards. To achieve the goal, several steps are required, which will be discussed later in detail.
    \item \textbf{Reputation Updating}. Reputation is an important indicator for each node in CycLedger. The higher a node's reputation, the more rewards he/she will get. After the above consensus phases, the reputation of all members in a committee are updated according to their votes. 
    \item \textbf{Referee Committee, Leaders and Partial Sets Selection}. We use a Proof-of-Work (PoW) process to figure out those nodes who will participate in the next round. After that, the referee committee and partial sets in round $r + 1$ are uniformly selected while leaders in the next round are selected concerning the updated reputation. At the same time, $C_R^r$ produces the randomness $R^{r + 1}$ for the next round.
    \item \textbf{Block Generation and Propagation}. After receiving all transactions, $C_R^r$ verifies and packs the legitimate ones as well as next round's randomness, next round's participants, nodes' updated reputation, and all members in $C_R^{r + 1}$, all leaders and partial set members into the block $B^r$, and releases it to all nodes. All nodes in a committee reach a consensus on the UTXOs they are in charge of after seeing $B^r$, and the leader will send them to $C_R^{r + 1}$.
\end{itemize}

    We roughly show how the protocol works when a transaction is submitted in Fig.~\ref{txToLedger}.

\begin{figure*}[!t]
        \centering
        \includegraphics[width=0.95\textwidth]{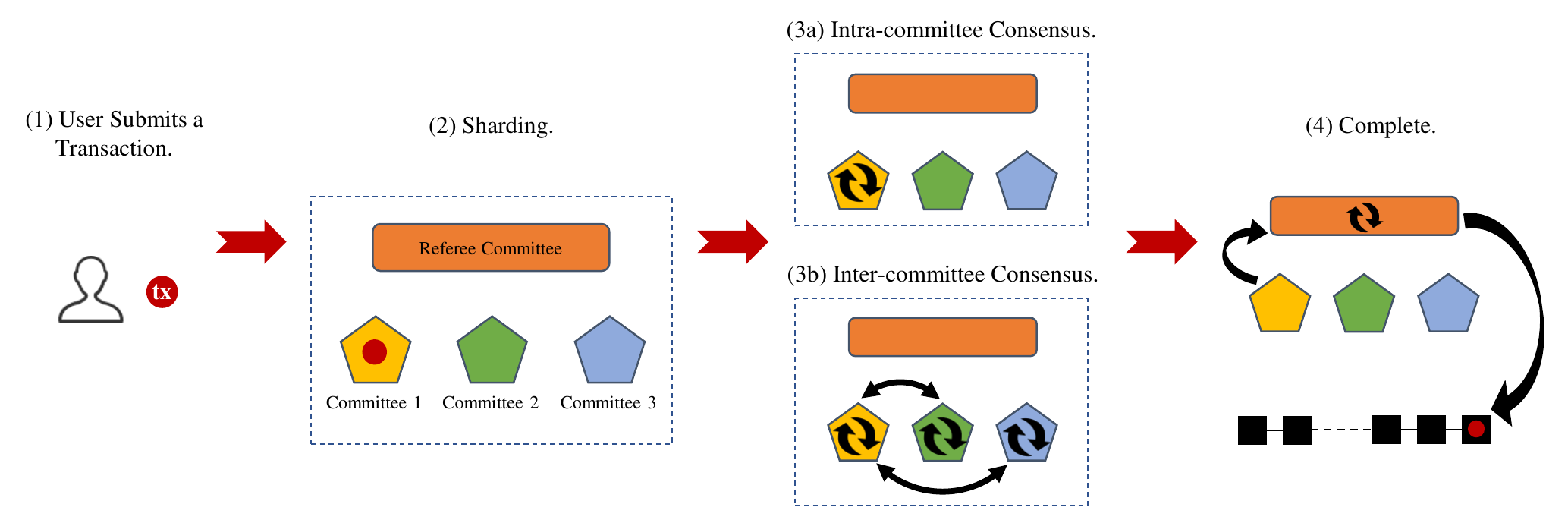}
        \caption{CycLedger architecture overview: (1) When someone submits a transaction (tx, for short), (2) the transaction will be sent to the corresponding shard, which is in the charge of committee 1. (3a) If the inputs and outputs of the transaction all belong to committee 1, they reach an intra-committee consensus on it. (3b) Otherwise, they reach an inter-committee consensus with the help of other committees. (4) After that, the verified transaction will be sent to the referee committee, who verifies and packs it into the block.}
        \label{txToLedger}
    \end{figure*}
\section{Main Protocol}\label{Protocol}

\subsection{Committee Configuration}

For simplicity, it is by default that all messages are sent authentically via the digital signature scheme throughout the protocol. We implicitly omit the signature verification in the description. Meanwhile, to avoid unnecessary tautology, a committee's leader and partial set members are referred to as \textit{key members} of the committee. Other validators in the committee are also known as \textit{common members}.

To begin with, we propose a cryptographic sortition mechanism using the Verifiable Random Function (VRF) scheme. Except for the pre-selected referee committee participants and key members, an undetermined node can find out the committee he/she belongs to via this method. Algorithm~\ref{Cryptographic Sortition} takes node's public/secret key pair $(PK, SK)$, round number $r$ and round's randomness $R^r$ as input, and returns a committee ID $id$ and a proof $\pi$ which certifies that the node belongs to $C_i^r$ in the round.

\begin{algorithm}[htb]
\caption{Cryptographic Sortition}
\label{Cryptographic Sortition}
\begin{algorithmic}[1]
\Procedure{Crypto\_Sort}{$PK, SK, r, R^r$}
    \State $<hash, \pi> \leftarrow VRF_{SK}(\mathsf{COMMON\_MEMBER}\| r \| R^r)$;
    \State $id \leftarrow hash\mod m$;\\
    \Return $(id, hash, \pi)$.
\EndProcedure
\end{algorithmic}
\end{algorithm}  

For a non-key member $i$:    
    \begin{enumerate}
        \item He/she determines his/her committee via Algorithm~\ref{Cryptographic Sortition}. 
        \item He/she sends his/her public key $PK_i$, address, VRF value $hash$ and VRF proof $\pi$ to the key members of the committee whose addresses are already shown in block $B^{r-1}$.
        \item When node $i$ receives a list from a key member, he/she delivers his/her public key, address, the hash value together with the proof to all unconnected committee members on the list.
        \item When $i$ receives a $PK_j$ from a validator $j$, first he/she checks if node $j$ is in the same committee with the given proof and forms the link if right. 
    \end{enumerate}

For a key member in committee $C_k$:
    \begin{enumerate}
        \item 
        He/she maintains a $<PK, address>$ list. Initially, the list contains the $<PK, address>$ pairs of all key members in the committee.
        \item
        When receiving a public key from node $j$, first he/she checks whether node $j$ belongs to the committee via the provided proof. If $j \in C^r_k$, he/she responds the current list back, and adds $<PK_j, address_j>$ into the list. 
    \end{enumerate}
    
Committees are formed in parallel as nodes independently execute the program shown in Algorithm~\ref{Comittee Configuration}. As a default, in all pseudocodes, the function \textit{BROADCAST} implicitly suggests that the message is multicast to all known members in the committee. We note that in expectation, each committee is consist of $c = O(\log^2 n)$ nodes. To avoid complex notations, we sometimes use $C_k$, $l_k$ and $C_R$ instead of $C_k^r$, $l_k^r$ and $C_R^r$ in algorithms when there is no ambiguity. Meanwhile, typically, we use $l$ to represent a leader, $pm$ to represent a partial set member when the specific committee is unimportant and $rm$ to represent a referee member.

\begin{algorithm}[htb]
\caption{Committee Configuration}
\label{Comittee Configuration}
\algrenewcommand\algorithmicwhile{\textbf{upon}}
\begin{algorithmic}[1]
\Procedure{Comm\_Config}{$r, B^{r - 1}$}\\
\vspace{0.35cm}
\textbf{For any key member $km$ in committee $C_k$:} 

\State $S \leftarrow \bigcup_{i \in \{l_k\} \cup C_{k, partial}} \{<PK_i, address_i>\}$; \vspace{0.15cm}
\State $\Psi \leftarrow \emptyset$;
\State $Q \leftarrow \mathsf{COMMON\_MEMBER} \| r \| R^r$;
\vspace{0.35cm}
\While{\textit{DELIVER}$(i\ |\ \mathsf{CONFIG}, <PK_i, address_i>, hash_i, \pi_i)$}
    \If{\textit{VRF\_VERIFY}$_{PK_i}(Q, hash_i, \pi_i) = TRUE$} 
        \State $S \leftarrow S \cup \{<PK_i, address_i>\}$;
        \State $\Psi \leftarrow \Psi \cup \{<hash_i, \pi_i>\}$;
        \State \textit{SEND}$(i\ |\ \mathsf{MEM\_LIST}, S)$;
    \EndIf
\EndWhile
\\
\vspace{0.35cm}
\textbf{For any non-key member $i$:}

\State $S \leftarrow \emptyset$;
\State $Q \leftarrow \mathsf{COMMON\_MEMBER} \| r \| R^r$;
\State $(id, hash, \pi) \leftarrow$ \textit{CRYPTO\_SORT}$(PK_i, SK_i, r, R^r)$;
\State $(l_{id}, C_{id, partial}) \leftarrow B^{r-1}$;
\State \textit{BROADCAST}$(\mathsf{CONFIG}, <PK_i, address_i>, hash, \pi)$;
\vspace{0.35cm}
\While{\textit{DELIVER}$(km\ |\ \mathsf{MEM\_LIST}, S')$}
    \State $S \leftarrow S \cup S'$;
    \State \textit{BROADCAST}$(\mathsf{MEMBER}, <PK_i, address_i>, hash, \pi)$;
\EndWhile
\vspace{0.35cm}
\While{\textit{DELIVER}$(j\ |\ \mathsf{MEMBER}, <PK_j, address_j>, hash_j, \pi_j)$}
    \If{\textit{VRF\_VERIFY}$_{PK_j}(Q, hash_j, \pi_j) = TRUE$}
        \State $S \leftarrow S \cup \{<PK_j, address_j>\}$;
    \EndIf
\EndWhile
\EndProcedure
\end{algorithmic}
\end{algorithm} 
    
\subsection{Semi-Commitment Exchanging}

This phase starts at a given time after the first phase starts. The recommended delay is $8\Delta$.

We rely on Algorithm~\ref{Consensus inside a Committee} which works efficiently to reach an agreement inside a committee. For simplicity, we omit the message verification procedure in the algorithm. A demonstration of its process is in Fig.~\ref{fig2}. Specifically, the algorithm contains three synchronous steps. In the first step, the leader multicasts the round number $r$, the original information $M$, the digest $H(M)$, and the sequence number of the message $sn$ to each group member with the tag $\mathsf{PROPOSE}$. Here, the sequence number $sn$ is unique and monotonically increasing over time. And the digest helps to mitigate the burden on the channel in later steps. When a committee member $i$ receives $M$ and $H(M)$, he/she verifies the correctness of the digest, checks the round number and the uniqueness of the sequence number, and broadcasts $<r, sn, H(M), i>$ with the tag $\mathsf{ECHO}$ as well as transmits the original $\mathsf{PROPOSE}$ message to all members in the committee. If $i$ receives the identical $\mathsf{ECHO}$ and transmitted messages from more than half validators in the committee as well as the corresponding $\mathsf{PROPOSE}$ message from the leader, he/she gossips $<r, sn, H(M), i>$ tagged by $\mathsf{CONFIRM}$ with all the authenticated $\mathsf{ECHO}$es he/she received back to the leader. If any non-faulty node notices that the leader is malicious (\textit{e.g.}, proposed different messages to different nodes), he/she informs all members of the committee immediately to stop the consensus process. A trustful partial set member then arouses a recovery procedure to evict the current leader and elect a new one. The recovery procedure will be explained in detail later. 

\begin{figure*}[!t]
\centering
\includegraphics[width=0.7\textwidth]{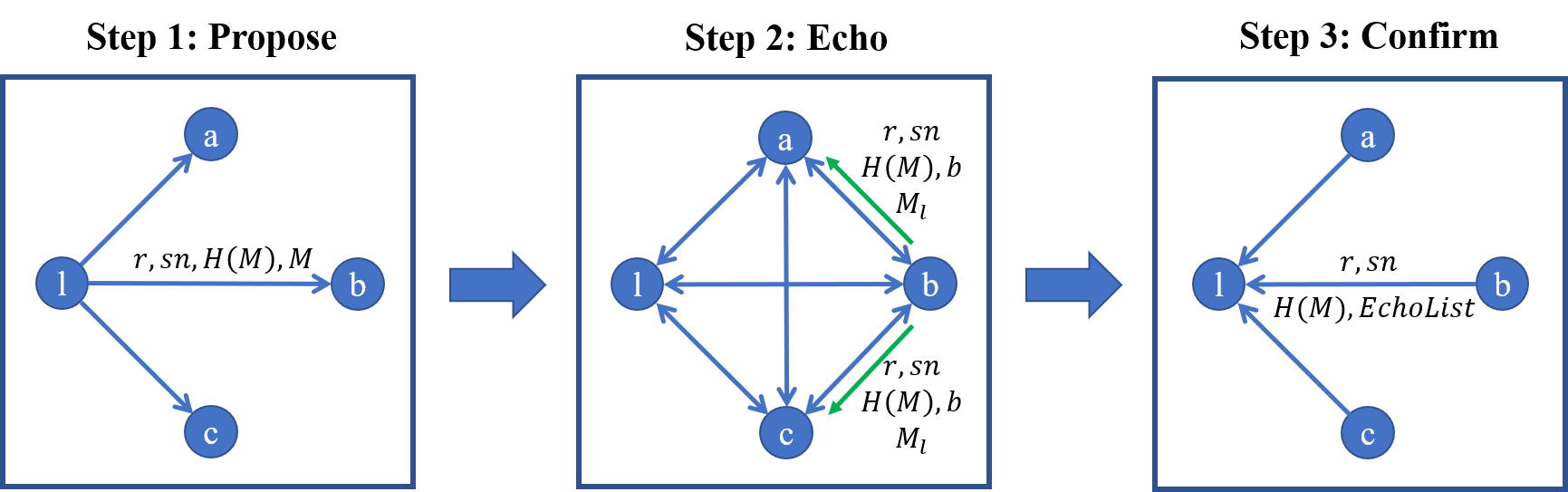}
\caption{A demonstration of inside-committee consensus.}
\label{fig2}
\end{figure*}

\begin{algorithm}[htb]
\caption{Inside-committee Consensus}
\label{Consensus inside a Committee}
\algrenewcommand\algorithmicwhile{\textbf{upon}}
\begin{algorithmic}[1]
\Procedure{Inside\_Consensus}{$r, sn, M$}\\
\vspace{0.35cm}
\textbf{For leader $l$:} 

\State $SigList \leftarrow \emptyset$;
\State \textit{BROADCAST}$(SIG_l<\mathsf{PROPOSE}, r, sn, H(M)>, M)$;
\vspace{0.35cm}
\While{\textit{DELIVER}$(i\ |\ SIG_i<\mathsf{CONFIRM}, r, sn, H(M), i>, EchoList_i)$}
    \State $M_{\mathsf{c}, i} \leftarrow SIG_i<\mathsf{CONFIRM}, r, sn, H(M), i>$;
    \State $v \leftarrow v + 1$;
    \State $SigList \leftarrow SigList \cup \{M_{\mathsf{c}, i}\}$;
    \If{$v > C / 2$}\Comment{$C :=$ the committee size.}\\
       \Return $SigList$;
    \EndIf
\EndWhile
\\
\vspace{0.35cm}
\textbf{For any member $i$ (including leader $l$):}
\State $EchoList \leftarrow \emptyset$;
\State $v \leftarrow 0$;
\vspace{0.35cm}
\While{\textit{DELIVER}$(l\ |\ SIG_l<\mathsf{PROPOSE}, r, sn, H(M)>, M)$}
    \State $M_{\mathsf{p}, l} \leftarrow SIG_l<\mathsf{PROPOSE}, r, sn, H(M)>$;
    \State \textit{BROADCAST}$(SIG_i<\mathsf{ECHO}, r, sn, H(M), i>, M_{\mathsf{p}, l})$;
\EndWhile
\vspace{0.35cm}
\While{\textit{DELIVER}$(j\ |\ SIG_j<\mathsf{ECHO}, r, sn, H(M), i>, M_{\mathsf{p}, l})$}
    \State $M_{\mathsf{e}, j} \leftarrow SIG_j<\mathsf{ECHO}, r, sn, H(M), j>$;
    \State $v \leftarrow v + 1$;
    \State $EchoList \leftarrow EchoList \cup \{M_{\mathsf{e}, j}\}$;
    \If{$v > C / 2$}
        \State $M_{\mathsf{c}, i} \leftarrow SIG_i<\mathsf{CONFIRM}, r, sn, H(M), i>$;
        \State \textit{SEND}$(l\ |\ M_{\mathsf{c}, i}, EchoList)$;
    \EndIf
\EndWhile
\EndProcedure
\end{algorithmic}
\end{algorithm}

Back to the phase, we only require the computational-binding property of a commitment scheme here. That is where the name "semi-commitment" comes from. The semi-commitment exchanging phase runs as follows. 
\begin{enumerate}
    \item 
    To start, each committee’s leader $l_k^r$ should unite the member list $S = \{PK^r_{k,1}, PK^r_{k,2}, \cdots\}$ from all key members in the committee, and compute the committee’s semi-commitment via the external hash function $H$:
    $$ SEMI\_COM_k^r = H(S). $$
    
    Then he/she multicasts it together with the member list to everyone in $C_R^r$. To prevent any means of cheating, he/she also delivers the latter with the belongingness certificate to the partial set $C^r_{k, partial}$.
    \item
    After every participant in $C_R^r$ receives semi-commitments from all committees, all trusty nodes in $C_R^r$ apply Algorithm~\ref{Consensus inside a Committee} to check i) all members in any list are registered; ii) all semi-commitments are valid. (To execute the algorithm, each node in $C_R^r$ is regarded as the leader.) They transmit the set of valid semi-commitments to all key members and expel the cheating leaders afterward.
    \item
    When a partial set member $c^r_{k, i}$ gets the semi-commitment $SEMI\_COM_k^r$ from $C_R^r$, he/she verifies whether his/her committee’s semi-commitment corresponds with the member list $S$ he/she receives from the leader. The list $S$ should be no smaller than the set he/she locally maintains. Meanwhile, the certificate should be valid. Once a truthful partial set member notices any mismatch, he/she may invoke the recovery procedure to evict the current leader.
\end{enumerate} 

We show the phase in Algorithm~\ref{Semi-Commitment Exchange}, however, the verifying process executed by partial set members is omitted.

\begin{algorithm}[htb]
\caption{Semi-commitment Exchange}
\label{Semi-Commitment Exchange}
\algrenewcommand\algorithmicwhile{\textbf{upon}}
\begin{algorithmic}[1]
\Procedure{Com\_Exchange}{$r, S, \Psi$}\\
\vspace{0.35cm}
\textbf{For leader $l_k$:} 

\State $SEMI\_COM \leftarrow H(S)$;
\State $ComList \leftarrow \bm{\vec{0}}$;
\State $ConfList \leftarrow \bm{0}$;
\For{$rm \in C_R$}
    \State \textit{SEND}$(rm\ |\ SIG_{l_k}<\mathsf{SEMI\_COM}, SEMI\_COM, r, S, k>)$;
\EndFor
\For{$pm \in C_{k, partial}$}
    \State \textit{SEND}$(pm\ |\ SIG_{l_k}<\mathsf{SEMI\_COM}, SEMI\_COM, r, S, \Psi>)$;
\EndFor
\vspace{0.35cm}
\While{\textit{DELIVER}$(rm\ |\ SIG_{rm}<\mathsf{SEMI\_COM}, SEMI\_COM_j, r, j, rm>)$}
    \State $ConfList[j][SEMI\_COM_j] \leftarrow ConfList[j][SEMI\_COM_j] + 1$;
    \If{$ConfList[j][SEMI\_COM_j] > |C_R| / 2$}
        \State $ComList[j] \leftarrow SEMI\_COM_j$;
    \EndIf
\EndWhile
\\
\vspace{0.35cm}
\textbf{For any referee member $rm$:}

\While{\textit{DELIVER}$(SIG_{l_k}<\mathsf{SEMI\_COM}, SEMI\_COM_k, r, S, k>)$}
    \State $SigList \leftarrow $\textit{Inside\_Consensus}$(r, <k, l_k>, <\mathsf{SEMI\_COM}, SEMI\_COM_k, k>)$
    \For{each leader $l$}
        \State \textit{SEND}$(l\ |\ SIG_{rm}<\mathsf{SEMI\_COM}, SEMI\_COM_k, r, k, rm>, SigList)$;
    \EndFor
\EndWhile
\EndProcedure
\end{algorithmic}
\end{algorithm} 

\subsection{Intra-committee Consensus}

The phase is a straight application of Algorithm~\ref{Consensus inside a Committee}.
\begin{enumerate}
    \item 
    In the beginning, after receiving some constant amount of transactions from external users, each leader $l^r_k$ creates a $TXList$ whose inputs are all within the shard they are in charge of. 
    \item The leader $l^r_k$ broadcasts his/her $TXList$ to everyone in the committee.
    \item
    After receiving $TXList$, every member votes listed transactions with $Yes$, $No$ or $Unknown$. When an honest node fails to judge a transaction within the given time, he/she should vote $Unknown$. Afterwhile, he/she forwards the voting list back to the leader.
    \item
    After the leader obtains all voting lists,  he/she picks up the set of transactions marked with a majority of $Yes$, which is named as $TXdecSET$.(Note that the collecting process should be within a certain time, \textit{e.g.} $6\Delta$ to avoid malicious nodes from indefinitely delaying. Those nodes who fail to reply in the period are deemed as voting $Unknown$ on all transactions. In the description of Algorithm~\ref{Intra-committee Consensus}, we tacitly approve that all members reply within the time limit.) Then he/she bales everyone's vote in a set named $VList$. The leader runs Algorithm~\ref{Consensus inside a Committee} to reach consensus on both $TXdecSET$ and $VList$. 
    \item
    Finally, the leader sends $TXdecSET$ with at least half members' certification to $C_R^r$.
\end{enumerate}

Readers can refer to Algortihm~\ref{Intra-committee Consensus} for the execution of the phase. 
\begin{algorithm}[htb]
\caption{Intra-committee Consensus}
\label{Intra-committee Consensus}
\algrenewcommand\algorithmicwhile{\textbf{upon}}
\begin{algorithmic}[1]
\Procedure{Intra\_Consensus}{$r, TXList$}\\
\vspace{0.35cm}
\textbf{For leader $l$:} 
\State $TXdecSET \leftarrow \emptyset$;
\State $VList \leftarrow \bm{\vec{0}}$;
\State $TXSigNum \leftarrow \bm{\vec{0}}$;
\State $ReplyNum \leftarrow 0$;
\State \textit{BROADCAST}$(\mathsf{TX\_LIST}, r, SIG_l<TXList>)$;
\vspace{0.35cm}
\While{\textit{DELIVER}$(i\ |\ \mathsf{VOTE}, r, SIG_i<VList_i>)$}
    \State $ReplyNum \leftarrow ReplyNum + 1$;
    \State $VList[i] \leftarrow VList_i$;    
    \For{$tx \in TXList$}
        \If{$VList_i[tx] = Yes$}
            \State $TXSigNum[tx] \leftarrow TXSigNum[tx]+1$;
            \If{$TXSigNum[tx] > C / 2$}\Comment{$C :=$ the committee size.}\\
                \State $TXdecSET \leftarrow TXdecSET \cup \{tx\}$;
            \EndIf
        \EndIf
    \EndFor
\EndWhile    
\vspace{0.35cm}
\While{$ReplyNum = C$}
    \State $SigList \leftarrow$ \textit{Inside\_Consensus}$(r, <l, \mathsf{INTRA}>, <TXdecSET, VList>)$;
    \For{each $rm \in C_R$}
        \State \textit{SEND}$(rm\ |\ SIG_l<\mathsf{INTRA}, r, TXdecSET, VList>)$;
    \EndFor
\EndWhile
\\
\vspace{0.35cm}
\textbf{For any member $i$ (including leader $l$):}

\While{\textit{DELIVER}$(l\ |\ \mathsf{TX\_LIST}, r, SIG_l<TXList>)$}
    \State $VList_i \leftarrow Vote(TXList)$;
    \State \textit{SEND}$(l\ |\ \mathsf{VOTE}, r, SIG_i<VList_i>)$;
\EndWhile

\EndProcedure
\end{algorithmic}
\end{algorithm}

\subsection{Inter-committee Consensus}

We use the semi-commitment scheme introduced before to ensure the security of cross-shard communication.

Consider those transactions whose inputs and outputs scatter across two shards, which are maintained respectively by $C_i^r$ and $C_j^r$. First of all, $l_i^r$ should reach a consensus on such transactions within $C_i^r$ by Algorithm~\ref{Consensus inside a Committee}. We call this list $TXList_{i, j}$. Then, the leader sends the consensus on $TXList_{i, j}$ as well as the member list to $l_j^r$ and $C_{j, partial}^r$. Note that in this process, a faulty leader cannot fabricate a consensus result concerning the semi-commitment.

Resembling the last phase, $C_j^r$ will reach an agreement on $TXList_{i, j}$. $l_j^r$ then sends the consensus result back to $l_i^r$.

\subsection{Reputation Updating}

Once reaching the agreement on a voting of a $TXList$ (either an inner-committee list or a cross-committee one), the leader grades every group member. As previously described, $VList$ contains $c$ marked votes, \textit{i.e.}, all members' opinion on the validity of listed transactions. On the vote, someone marks $Yes$ for those transactions he/she agrees, $No$ for disagreed and $Unknown$ for the left. Facing the votes of all members, the leader scores each member according to the proximity between his/her opinion and the final decision.

Let $+1$, $-1$ and $0$ represent $Yes$, $No$ and $Unknown$ respectively. Suppose the amount of transactions to be determined is $D$. Then there is a $D$-dimensional vector space and each transaction represents an axe of the space. From this aspect, a vote is seen as a vector in this space. Let $\vec{v_i}=\{v_{i,k}|k=1,2, \cdots, D\}$ denote the vote of member $i$, where $v_{i,k}$ is the member $i$'s opinion on the $k^{th}$ transaction. We use the cosine of the angle between two vectors to measure the proximity of two corresponding opinions. In other words, we define the $score$ of one node as the cosine similarity between its voting vector and the resulting vector, which is determined by the majority algorithm and denoted by $\vec{u}=\{u_{k}|k=1,2, \cdots, D\}$. Let $s_{i}$ denote the member $i$'s score. We have
\begin{equation}
s_{i} = cos(\vec{v_i},\vec{u}) 
= \frac{\sum_{k=1}^{D}v_{i,k} u_{k}}{\sqrt{\sum_{k=1}^{D}v_{i,k}^2}\sqrt{\sum_{k=1}^{D}u_{k}^2}} \in [-1,1].    \label{score}
\end{equation}

Owing to the random generation of each committee and the design of Algorithm~\ref{Consensus inside a Committee}, the majority result reflects truthful members' views.  Here, the main idea is, the closer one's opinion stands with the consensus, the higher the grade he/she will get. Specifically, if a member has the same answers with the consentaneous results, he/she would be rewarded the highest score, \textit{i.e.}, $+1$. On the contrary, if a member replies with completely opposing opinions, he/she will face a loss of $-1$ in reputation.

After calculating all scores, the leader assembles them into a $ScoreList = \{s_{1}, s_{2}, \cdots, s_{c}\}$. Then he/she broadcasts $ScoreList$ and the original $VList$ to all members, waiting for the consensus by Algorithm~\ref{Consensus inside a Committee}. In this process, each non-faulty member should sign on the $ScoreList$. If successful, the leader sends the agreement to $C_R^r$, together with relevant certification. Then $C_R^r$ updates their reputation by simply adding the listed score.
    
\subsection{Referee Committee, Leaders and Partial Sets Selection}

Participants in $C_R^r$ distributedly generate next round's seed $R^{r+1}$ via a random beacon generator. Here, the SCRAPE~\cite{PVSS2} scheme is preferred as it guarantees the pseudorandomness and unbiasedness of the seed even when the adversary takes control of almost half nodes. The nodes who want to participate in the next round need to solve a PoW puzzle in advance. The difficulty of the puzzle is appropriate and equal to everyone. Upon solving, a node is supposed to deliver the solution to the referee committee $C_R^r$ and the latter will record his/her identity. As a result, by the end of the round $r$, $C_R^r$ is aware of all next round's participants $P^{r+1}$. $C_R^r$ chooses $m$ nodes with the highest reputation as new leaders in round $r+1$. Afterwhile, based on randomness $R^r$, $C_R^r$ determines the next referee committee $C_R^{r+1}$ and partial sets in the next round $\{C_{1, partial}^{r+1},\cdots,C_{m, partial}^{r+1}\}$. For example, a member in $C_R^r$ can see if the following inequality holds for a node $i$:
    $$ H(r+1 \| R^r \| PK_i \| role) \leq d^r(role) $$
where $d(\cdot)$ is a difficulty function and $role$ is an optional string including $\mathsf{REFEREE\_COMMITTEE\_MEMBER}$ or $\mathsf{PARTIAL\_SET\_MEMBER}$. A new $d^r(role)$ for either $role$ can be proposed every several rounds as the number of nodes in the network changes. If a node is selected as a partial set member, a member in $C_R^r$ can calculate $H(r+1 \| R^r \| PK_i \| \mathsf{PARTIAL\_SET\_MEMBER}) \mod m$ to determine the committee he/she belongs to.
    
\subsection{Block Generation and Propagation}

Each time $C_R^r$ members receive a $TXdecSET$ from a leader, they check the validity of its signature. By the end of the round, the referee committee comes to an agreement using Algorithm~\ref{Consensus inside a Committee} on the set of valid $TXdecSET$s and pack them up, together with all participants of next round $S^{r+1}$, their reputations $W^{r+1}$, the elected referee committee $C_R^{r+1}$, leaders and partial sets as a \textit{block} $B^r$. By releasing it to the whole network, all nodes could obtain the contained information. After viewing the proposed block $B^r$, each committee member traverses all transactions packed in it. In this process, members delete the used ones from their local \textit{UTXO Lists} and append the newly generated outputs that they are responsible for. Meanwhile, limited by the package size or other reasons, there may be some unpacked valid transactions within each committee, which form a \textit{Remaining TX List}. The leader of each committee runs Algorithm~\ref{Consensus inside a Committee} to reach a consensus on the final \textit{UTXOs List} and \textit{Remaining TX List}, and sends them to $C_R^r$. Next, $C_R^r$ binds these lists with the \textit{Committee ID} and forwards them to the corresponding new partial sets. Up to this point, all committees have finished their tasks.
    
Afterward, a node obtains part of the transaction fee based on his/her reputation. The sum of all nodes' revenue equals the fee of all transactions admitted in this round. Considering that one's reputation may be negative, we first map the reputation to a positive number using a monotone function $g(\cdot)$. Then rewards are distributed proportionally to the mapped value. Specifically, the function $g(\cdot)$ is designed as follows. The reputation of $C_R^r$ will be updated by $C_R^{r+1}$ in the next round.
    
    \begin{equation}
        g(x) = \left\{
            \begin{array}{rcl}
                e^x,       && {x \leq 0;}\\
                1 + \ln(x+1),       && {x > 0.}
            \end{array} 
        \right.
        \label{eq:monotone}
    \end{equation}     
    
    \begin{figure}[!t]
        \centering
        \includegraphics[width=0.4\textwidth]{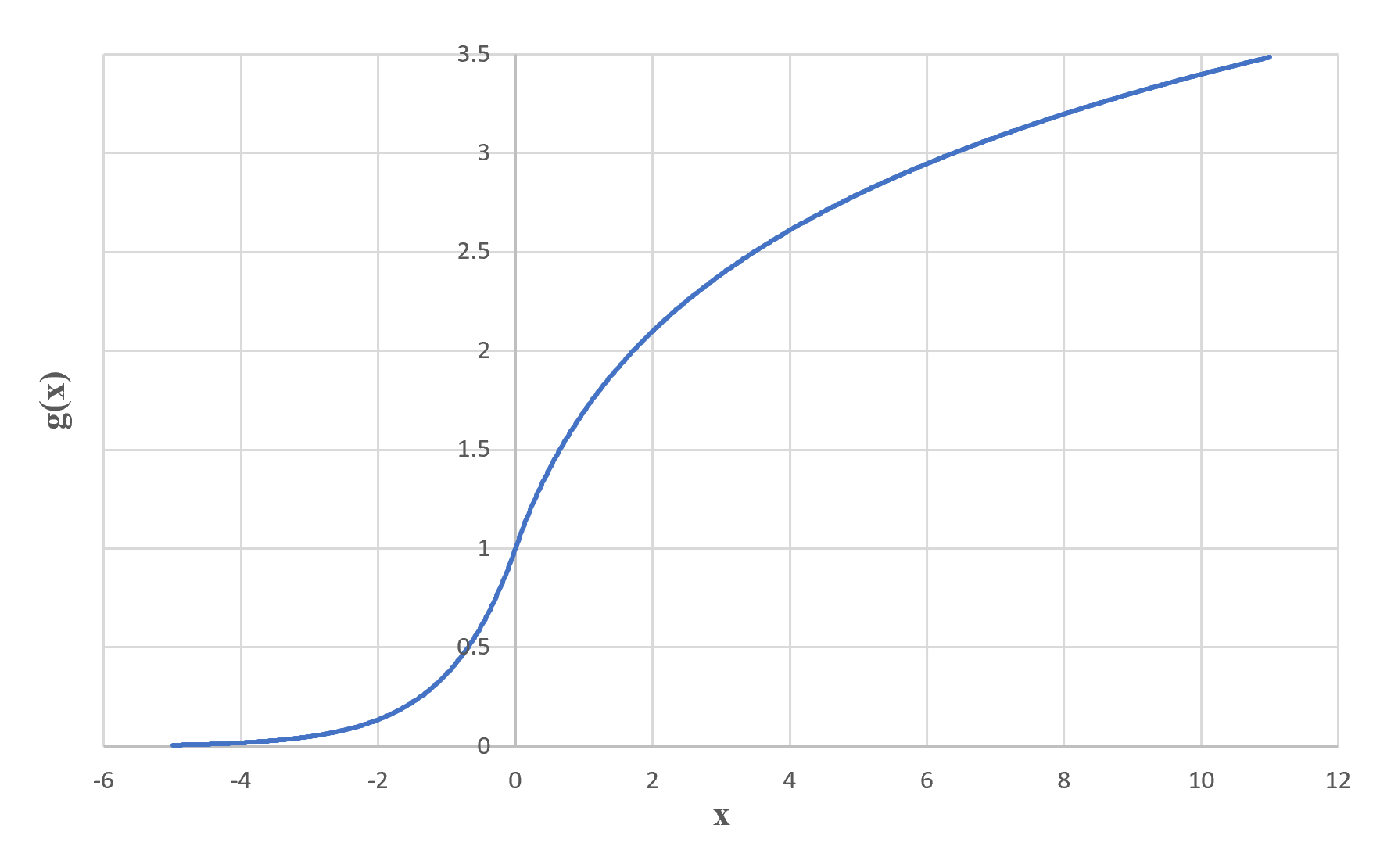}
        \caption{The monotone function $g(x)$ mapping reputation to a positive number.}
        \label{fig:monotone}
    \end{figure}
    
The monotone increasing function and proportional distribution ensure that whoever works more gets more. According to \eqref{eq:monotone}, $g(0)=1$. Thus, nodes whose reputation is zero (\textit{e.g.}, nodes who always vote \textit{Unknown}) could still get little rewards. By contrast, the negative reputation is mapped to near zero, which means the corresponding nodes can hardly obtain revenue. Under this reward mechanism, it is better to do nothing rather than do something bad, thus encouraging the malicious nodes to do right.
\section{Security Analysis}\label{Security}
In this section, we provide a comprehensive discussion on the security of our protocol, showing that CycLedger is highly secure with overwhelming probability.

\subsection{Security on Randomness}

In CycLedger, we apply the SCRAPE scheme~\cite{PVSS2} within $C_R$ to distributedly generate a random string. SCRAPE guarantees that as long as the majority of nodes in $C_R$ are honest, the output random string is pseudorandom and unpredictable. At the same time, no leader is required in the execution of SCRAPE. This feature suits the construction of the referee committee well as $C_R$ is the only committee without a leader. As one can see in the following analysis, each committee, including $C_R$, has more than half of non-faulty nodes with high probability, hence, we assert that the randomness produced by SCRAPE with $C_R$ is reliable.

\subsection{Security on Committee Configuration}

\label{Security on Committee Configuration}

We say a committee is secure when more than half of nodes are non-faulty. Recall that committees are formed uniformly except leaders. Let $X$ denote the number of malicious nodes in a committee, and $c$ be the expected committee size. We consider the tail bound of hypergeometric distribution which gives the following result:

\begin{equation}
    \Pr[X \geq \frac{c}{2}] = \sum_{x = \frac{c}{2}}^c \frac{\binom{t}{x}\binom{n - t}{c - x}}{\binom{n}{c}} \leq e^{-D(\frac{1}{2}||f)c},
\end{equation}
where $D(\cdot||\cdot)$ is the Kullback-Leibler divergence. Here $t < \frac{n}{3}$ and $f < \frac{1}{3} + \frac{1}{c}$, thus,

\begin{equation}
    \Pr[X \geq \frac{c}{2}] \leq e^{-\frac{c}{12}}.
    \label{eq:committee_error}
\end{equation}

When the expected committee size is $c = O(\log^2n)$, we derive that the probability that a committee is insecure is less than $n^{\frac{-\log n}{12}}$, which is negligible of $n$.

Fig.~\ref{fig_pro_fail} visualizes \eqref{eq:committee_error}. Namely, it shows the probability of failure calculated using the hypergeometric distribution to uniformly sample a committee when the population of the whole network is 2000. The amount of malicious nodes is 666, exactly less than one-third of the size of the network.

\begin{figure}[ht]
    \centering
    \includegraphics[width=0.4\textwidth]{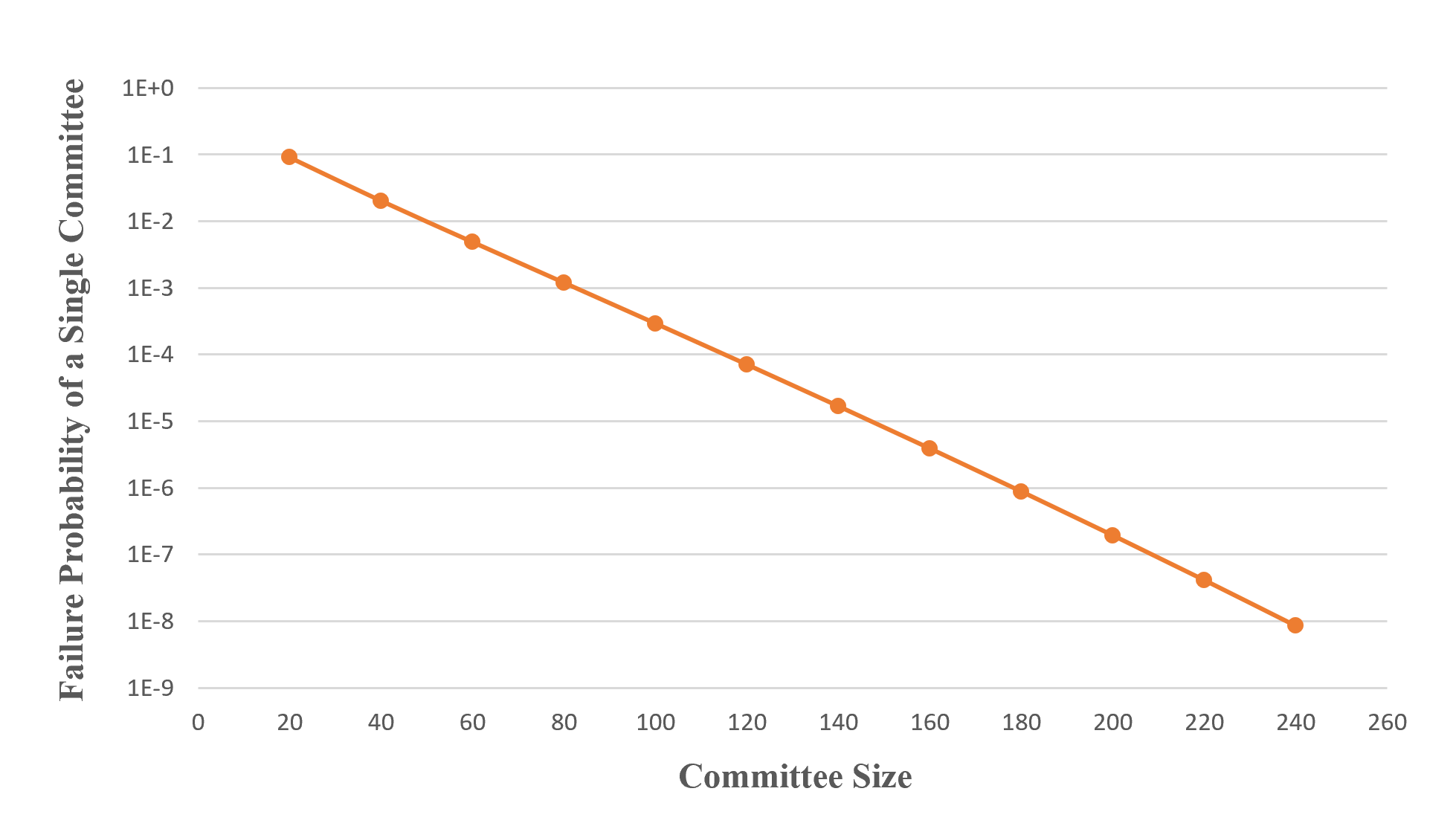}
    \caption{Probability of failure in sampling one committee from a population of 2000 nodes in CycLedger. The amount of malicious nodes is set to 666.}
    \label{fig_pro_fail}
\end{figure}

Particularly, when $c = 240$, the error probability for a single committee is less than $2.1 \times 10^{-9}$. Applying union bound, when $m$ is less than 20, the error probability is no more than $5 \times 10 ^{-8}$.

\subsection{Security on Partial Sets}

We say a partial set is secure when at least one node in the set is honest. As no more than 1/3 validators are faulty, when the size of the partial set is set to 40, the probability that a partial set is insecure at most:
    $$ (\frac{1}{3}) ^ {40} < 8 \times 10^{-20}. $$

Associated with union bound, when the number of committees is 20, the probability that at least one partial set is insecure is no more than $2 \times 10^{-18}$.

\subsection{Security on Semi-Commitments}

Recall that we use the hash of the member list as a semi-commitment of a committee. To start with, we show that $SEMI\_COM_k^r = H(S)$ satisfies the computational binding property, where $S$ is the member list.
\begin{lemma}
When $H$ is modeled as a collision-resistant hash function (CRHF), $SEMI\_COM_k^r$ satisfies the computational binding property, \textit{i.e.}, once the semi-commitment is released,  a probabilistic polynomial-time malicious leader cannot forge another member list which corresponds to the same semi-commitment with non-negligible probability.
    \label{claim:release_com}
\end{lemma}

\begin{proof}
The lemma can be directly derived from the collision-resistance property of a CRHF.
\end{proof}

With the above lemma, we come up with the following theorem.
\begin{theorem}
A malicious leader cannot deceive a trustful leader by forging a member list of his/her committee as long as $C_R$ has an honest majority.
\end{theorem}

\begin{proof}
There are only two opportunities when a betrayer can lie to a loyalist on the member list. The first chance is when he/she tries to broadcast the semi-commitment. However, because all members in $C_R$, as well as the partial set members of the committee, see the exact member list, any false hash on the list will be perceived. Thus the liar will be detected. The other chance for the malicious leader is when the semi-commitment is revealed. However, according to the Lemma~\ref{claim:release_com}, misleading behavior will not take effect in this phase.
\end{proof}

We emphasize that the hiding property of a commitment is not necessary for our protocol. For a vicious leader, even if he/she figures out the members of a committee in the semi-commitment exchanging phase, he/she can do nothing under our threat model as the adversary cannot get control of a trusty node immediately.

Here we introduce the leader re-selection procedure in CycLedger. The program is invoked when an honest partial set member notices that his/her leader is malicious or any participant of $C_R$ notice that some leader is vicious. In the semi-commitment case, the event happens when a loyal party ($C_R$ or a partial set member) discovers inconsistency between the member list he/she owns and the leader claims.

If a partial set member wants to accuse his/her leader, he/she would broadcast his/her \textit{witness} to all members in the committee and ask them to vote on the impeachment. Here, a witness is a pair of messages $W = (m_l, m_0)$ where $m_l$ should be sent and signed by the leader. We say a witness is \textit{valid} if and only if the pair can derive dishonest behaviors of the leader. \textit{E.g.}, $m_l$ be the member list that the leader sends, and $m_0$ be the semi-commitment of the committee where $m_0 \neq H(m_l)$. If the proposal is approved by more than half of the validators, the prosecutor will forward the voting result as well as his/her witness to everyone in the referee committee.

When any node in $C_R$ receives a witness $W$ and a signature list $Cert$ approving the prosecution from a partial set member $pm$ from committee $C_k$, he/she starts Algorithm~\ref{Leader Re-selection} to re-select a committee leader.

\begin{algorithm}[ht]
    \caption{Leader Re-selection}
    \label{Leader Re-selection}
    \algrenewcommand\algorithmicwhile{\textbf{upon}}
    \begin{algorithmic}[1]
    \Procedure{Re-selection}{$r, pm, SIG_{pm}<k, W, Cert, pm>$}\\
    \vspace{0.35cm}
    \textbf{For any referee member $rm$:}
    \State $SigList \leftarrow $\textit{Inside\_Consensus}$(r, <k, pm>, SIG_{pm}<k, W, Cert, pm>)$;
    \For{each $i \in C_k$}
        \State \textit{Send}$(i\ |\ SIG_{rm}<\mathsf{NEW}, pm>, SigList)$;
    \EndFor
    \EndProcedure
\end{algorithmic} 
\end{algorithm}  

Fig.~\ref{fig3} shows the above process. Afterwhile, the new leader needs to make a new semi-commitment of the committee via the semi-commitment exchanging protocol (Algorithm~\ref{Semi-Commitment Exchange}). When a participant of $C_R$ receives the new semi-commitment, he/she informs every committee leader the new semi-commitment and leader's address, so that cross-shard transaction handling may start safely.

Now we claim that the given procedure is both complete and sound:
\begin{claim}\label{claim_com}
A faulty leader is always detected and thus evicted via the leader re-selection procedure, as long as $C_R$ has an honest majority.
\end{claim}

\begin{proof}
Note that there is at least one credible node in any partial set. Therefore, as a leader's action is always monitored by the partial set during the execution of the whole protocol (see further analysis in later subsections), any irregular behavior from the leader will be detected and a witness will be inevitably grasped by the non-faulty partial set member. At the same time, as the evidence is signed by the "criminal", no erroneous judgment will occur.
\end{proof}

\begin{claim}\label{claim_so}
A trustful leader will never be framed up by a faulty partial set member, as long as $C_R$ has an honest majority.
\end{claim}

\begin{proof}
We mention that a witness is valid if and only if the first part of it is a message signed by the leader. For the security of the digital signature scheme, a vicious partial set member cannot counterfeit a shred of evidence as the latter has to be a leader's signed message. Therefore, a loyal leader will never be unjustly accused.
\end{proof}

As proved above, the probability that more than half members in $C_R$ are faulty is negligible. Thus, we claim that our recovering procedure remains complete and sound with high probability.

\begin{figure*}[!t]
\centering
\includegraphics[width=0.8\textwidth]{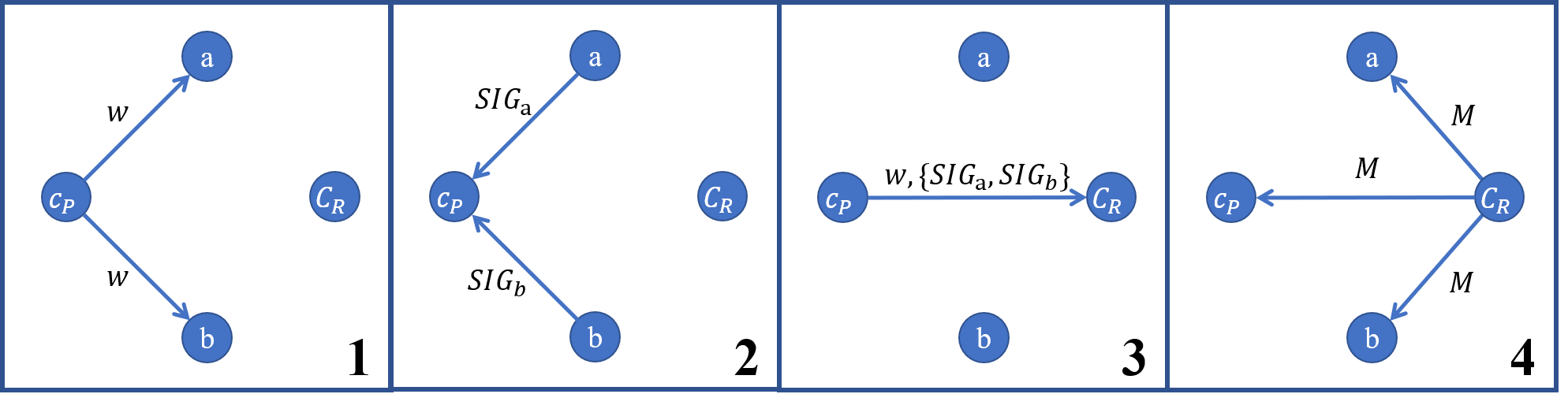}
\caption{This figure shows an example of the reporting mechanism and leader re-selection, where $C_R$ is the \textit{Referee Committee}, $c_p$ is a partial set member, $a$ and $b$ are two committee members.}
\label{fig3}
\end{figure*}

\subsection{Security on Intra-committee Consensus}

Resembling Claim~\ref{claim_com} and~\ref{claim_so}, We assert that a faulty leader is always detected and expelled in the intra-committee consensus phase.

\begin{theorem}
In Algorithm~\ref{Intra-committee Consensus}, a faulty leader can always be detected, meanwhile, malicious members can never calumniate a non-faulty leader.
\end{theorem}

\subsection{Security on Inter-committee Consensus}

Finally, we show that cross-shard transactions are safely processed by our protocol.

\begin{lemma}\label{claim_inter_com}
A malicious leader who tries to imitate or conceal some cross-shard transactions is always detected by a trustful partial set member and thus evicted via the leader re-selection procedure. 
\end{lemma}

\begin{proof}
We prove this lemma conditioned on the fact that the messages are unforgeable. If the malicious leader imitates or conceals some cross-shard transactions, a trusty partial set member can challenge the leader’s honesty by checking signatures from members of the departing committee. 
\end{proof}

\begin{lemma}\label{claim_inter-so}
A malicious leader can never frame up a trustful leader by misleading partial set members, as long as the delay of communication between shards is $\Gamma$.
\end{lemma}

\begin{proof}
Notice that if the faulty leader tries to frame up a credible leader, the only way is to sending a message to the partial set members, meanwhile do not send anything to the honest leader, misleading that the non-faulty leader hides all transactions. However, if a non-faulty partial set member does not receive transactions from his/her leader after $2\Gamma$ time since he/she receives the transactions set from another committee leader, he/she can send the transactions set to his/her leader and continues running consensus protocol.
\end{proof}

Combining the above two lemmas, we give the following theorem, which shows that our protocol guarantees security when processing cross-shard transactions.

\begin{theorem}
In the inter-committee consensus phase, a faulty leader can always be detected, in the meantime, non-faulty leaders can never be framed up.
\end{theorem}

\begin{proof}
This theorem can be derived by Lemma~\ref{claim_inter_com} and~\ref{claim_inter-so}.
\end{proof}

\section{Performance Analysis}\label{Performance}

We claim that our protocol is efficient in intra-shard and cross-shard transactions processing. 

Recall that we use $n$ to denote the number of nodes in the network, $m$ to denote the number of committees and $c$ to denote the expected amount of participants in a committee.

\subsection{Complexity of Committee Configuration}

In this phase, all members in any committee except $C_R$ will recognize each other, which imposes a communication, computation and storage complexity of $O(c)$ for all common members. For leaders and partial set members, the communication complexity is multiplied by $c$ as each of them has to deliver $i$ pieces of information to the $i^{th}$ applying participant.

\subsection{Complexity of Commitment Exchanging}

For any leader, he/she is obliged to produce the commitment of his/her committee and note down commitments from all other committees. Thus, the computation complexity is $O(c)$ while the storage complexity is $O(m)$. However, members in $C_R$ has to suffer from the huge transportation overhead of $O(m^2)$ as they have to propagate all commitments to all committees as intermediaries.

\subsection{Complexity of Reaching Intra-committee Consensus}

To reach a consensus on a set of transactions with constant size, one in a committee has to broadcast the set with his/her signature to all members in a committee, which causes a communication complexity of $O(c)$. For key members, they must store the information with at least half of the members' certification to acknowledge the set, hence, a storage complexity of $O(c)$ is induced. At the same time, common members only need to keep their own opinion in reserve. 

\subsection{Complexity of Block Generation and Propagation}

To propose a block with size $O(n)$ and broadcast it to all committees, an $O(mn)$ communicating burden and an $O(n)$ storage overhead are inevitable for any participant in $C_R$. However, we point out that the expense also exists in almost all previous protocols. At the same time, for any other attendee, the storage complexity is just $O(c)$, as he/she only needs to maintain the participants, transactions, and UTXOs within the committee he/she belongs to.

We summarize the theoretical analysis of the performance of CycLedger in Table~\ref{Complexity}.

\begin{table*}[!t]
    \renewcommand{\arraystretch}{1.35}
    \caption{Communication, Computation \& Storage Complexity of CycLedger}
    \label{Complexity}
    \centering
    \begin{threeparttable}
        \begin{tabular}{|c|c|c|c|c|}
           \hline
             Communication \& Computation / Storage Complexity & Common Members & Leaders \& Partial Set Members & $C_R$ Members \\
          \hline
             Committee Configuration & $O(c) / O(c)$ & $O(c^2) / O(c^2)$ & - \\
            \hline
            Semi-Commitment Exchanging & - & $O(c) / O(m)$ & $O(m^2) / O(m)$ \\
            \hline
            Intra-committee Consensus & $O(c) / O(1)$ & $O(c) / O(c)$ & $O(n) / O(n)$  \\
            \hline
            Inter-committee Consensus & $O(m) / O(1)$ & $O(n) / O(1)$ & $O(n) / O(n)$  \\
            \hline
            Reputation Updating & $O(c) / O(1)$ & $O(c) / O(c)$ & $O(n) / O(n)$  \\
            \hline
            Key Member Selection & - & - & $O(n) / O(n)$  \\
            \hline
            Block Generation \& Propagation & $O(m) / O(c)$ & $O(n) / O(c)$ & $O(mn) / O(n)$  \\
            \hline
        \end{tabular}
        \begin{tablenotes}
           \footnotesize
           \item[1] $n$: total amount of nodes in the network\ $m$: amount of committees\ $c$: amount of nodes in a committee, we mention again that $n = mc$.
        \end{tablenotes}
    \end{threeparttable}
\end{table*}
\section{Incentive Analysis}\label{Incentive}

Besides security, reputation is also a highlight of CycLedger. There are two main problems when introducing reputation to our protocol: what reputation reflects and how reputation works. In this section, we give a discussion on the incentive of our protocol, by analyzing these two problems. 

\subsection{Incentive on Reputation}

In general, reputation is expected to reflect the honest computational resources of each node in CycLedger.

The basic task of nodes in each committee, in short, is to give opinions on the validation of requested transactions. One's reputation is a reflection of his/her behaviors. For a newly joined node, based on his/her blank work experience, the reputation will start from zero. However, as long as he/she begins to work, the reputation matches his/her behavior. Specifically, for an honest node with more computing resources, he/she could make more correct judgments on the validation of transactions, thus, earning a higher reputation. While for the malicious nodes, reputation is closely related to his/her evilness, or in other words, the honest computational resources he/she contributes.  

One's reputation determines his/her profit in each round. Nodes with a higher reputation get more payment after a new block is successfully generated. Owing to the scoring mechanism and reward mechanism, reputation provides enough incentive for nodes to work honestly and as hard as they can. In this way, reputation is considered as a reflection of the trustworthy computational resources one node contributes.

In CycLedger, leaders have a higher workload compared with other members. With this in mind, we directly choose nodes with the highest reputation as leaders in each round, thus to enhance the performance and throughput of CycLedger even further. In return, leaders obtain some extra reputation as a bonus for their hard work, which turns into higher revenue. Therefore, leaders will have enough incentive to behave trustfully.

Certainly, a high reputation is not once and for all. Recalling the proportional design in our reward mechanism, one's revenue depends on the relative value of his/her reputation. Thus for each node, not to advance is to go back. So it appears that the best possible strategy for a node is, to use all his/her computational resources to work within rules.

\subsection{Punishment on Reputation}
We have been focusing on how the reputation motivates nodes to work hard. In this part, we will discuss what if someone, especially a leader, breaks the rules.

Intuitively, the scoring mechanism can be translated as, awarding points for right answers, deducting marks for wrong answers and doing nothing for an $Unknown$ reply. Thus when giving wrong votes, intentionally or unintentionally, the node will face the corresponding decline in reputation. That is, the scoring mechanism itself contains the punishment on reputation.

Moreover, a leader who violates the protocol faces a more serious penalty. Once a leader is confirmed to commit a fault by the referee committee, his/her reputation will be decreased to the cube root. Recall that all leaders are those nodes with the highest reputation. We believe that the reputation of each leader is larger than 0, including malicious ones. Combining this punishment with \eqref{eq:monotone}, the mapped value, which is closely related to his/her revenue, will reduce to about one-third of the original mapped value. That is, the higher the reputation a leader has, the stronger the punishment he/she will suffer when he/she behaves maliciously.
\section{Future Works}\label{FutureWorks}
In this section, we introduce two skills that may enhance the efficiency of our protocol.

\subsection{Excluding Low-value Transactions through Extra Communication}

According to our protocol, if $l_i^r$ wants to pack transactions that are related to UTXOs managed by $C_j^r$, $l_i^r$ should pack up those transactions while nodes in $C_i^r$ should run Algorithm~\ref{Consensus inside a Committee} to generate a package $PACK$. Then nodes in $C_j^r$ should also run Algorithm~\ref{Consensus inside a Committee} again to confirm if $PACK$ is valid. However, in some situations, most transactions in $PACK$ may be invalid, for example, when the network suffers a Denial-of-Service (DoS) attack. In this case, this interaction process may be a waste of computational resources.

We hope that transactions chosen by $l_i^r$ have a high probability to be accepted by $C_j^r$ to enhance efficiency. One practical way is that leaders can communicate with each other before sending packages. For instance, $l_i^r$ can enquire $l_j^r$ which transactions are valid, and receives a preference directly from $l_j^r$ rather than the agreement of $C_j^r$. This extra step of communication reduces the number of invalid transactions being packaged as long as both leaders trust each other. 
 
We can set up a mechanism to mitigate the possibility that either leader lies. To achieve this, $l_i^r$ can record the response of $l_j^r$ and then generate a packet by Algorithm~\ref{Consensus inside a Committee}, including all transactions mentioned by $l_j^r$. Thus, if $l_j^r$ lies to $l_i^r$ on the validity of a transaction, he/she gets punished, such as a reduction on reputation.

\subsection{Parallelizing Block Generation}

In our protocol, $C_R$ is in charge of proposing a block at the end of a round. This causes a huge connection burden on the referee committee. To boost efficiency, we can have each committee broadcast the block. In detail, after receiving enough authentication on a certain set of transactions from committee participants, a leader can forward the set to $C_R$ for verification. After obtaining permission from the referee committee, the leader can immediately broadcast the block to the whole network. Applying this change to our protocol, we can adapt the schedule proposed by \cite{Omni} in our protocol. We call two transactions are relevant if they follow one of the following properties:
\begin{enumerate}
    \item They use the same UTXO as input;
    \item One transaction spends the output of the other one.
\end{enumerate}

Observe that those irrelevant transactions can be processed in parallel. Thus, a committee can sequentially produce and broadcast blocks within a round. As a result, two transactions satisfying the second property above may both be accepted in the same round. This event never happens in our original protocol as at least one of them will be regarded as illegal. Thus, by using this mechanism, we can enhance the efficiency and reliability of our protocol.
\section{Conclusion}\label{Conclusion}

We present CycLedger, a 1/3-resilient sharding-based distributed ledger protocol with scalability, reliable safety, and incentives.

By splitting nodes into parallel committees, we maximize the utilization of the computational resource, bringing high throughput to our protocol. We enable users to trade safely by introducing semi-commitments among committees and a recovery procedure to detect and evict faulty leaders. By evaluating each validator's behavior explicitly, our protocol has a considerable incentive for nodes to follow the instructions. At the same time, the reputation mechanism helps CycLedger locate those nodes with higher trusty computational power. By assigning them to high-workload positions, we further enhance the efficiency of the protocol. Finally, our analysis demonstrates that CycLedger has a nice performance and can provide striking security.

\bibliographystyle{abbrv} 
\bibliography{References}
\end{document}